\documentclass[preprint,3p,authoryear,11pt,fleqn]{elsarticle}
\usepackage{amsmath,amsthm}
\usepackage{amssymb}
\usepackage{dsfont}
\usepackage{amsfonts}
\usepackage{enumerate}
\usepackage{mathrsfs}
\usepackage{multirow}
\usepackage{float}
\usepackage{graphicx}
\usepackage{color}
\newtheorem{applem}{Lemma A\!\!} 
\usepackage{abbreviationsSSC-ext}

\theoremstyle{definition}\newtheorem{remark}{Remark}

\begin{document}
\begin{frontmatter}

\title{Asymptotics of the two-stage spatial sign correlation}

\author[AD]{Alexander D{\"u}rre\corref{ad.cor}}
\author[DV]{Daniel Vogel}
\address[AD]{Fakult\"at Statistik, Technische Universit\"at Dortmund, 44221 Dortmund, Germany}
\address[DV]{Institute for Complex Systems and Mathematical Biology, University of Aberdeen, Aberdeen AB24 3UE, United Kingdom}
\cortext[ad.cor]{corresponding author: alexander.duerre@udo.edu, +49-231-755-4288}

\begin{abstract}
The spatial sign correlation \citep*{duerre:vogel:fried:2015} is a highly robust and easy-to-compute, bivariate correlation estimator based on the spatial sign covariance matrix. 
Since the estimator is inefficient when the marginal scales strongly differ, a two-stage version was proposed. In the first step, the observations are marginally standardized by means of a robust scale estimator, and in the second step, the spatial sign correlation of the thus transformed data set is computed. 
\citet{duerre:vogel:fried:2015} give some evidence that the asymptotic distribution of the two-stage estimator equals that of the spatial sign correlation at equal marginal scales by comparing their influence functions and presenting simulation results, but give no formal proof. In the present paper, we close this gap and establish the asymptotic normality of the two-stage spatial sign correlation and compute its asymptotic variance for elliptical population distributions. We further derive a variance-stabilizing transformation, similar to Fisher's $z$-transform, and numerically compare the small-sample coverage probabilities of several confidence intervals.
\end{abstract}

\begin{keyword}
scale estimator\sep spatial sign correlation\sep spatial sign covariance matrix  

\MSC[2010] 62H12 \sep 62G05 \sep 62G35
\end{keyword}

\end{frontmatter}



\renewcommand{\thefootnote}{\fnsymbol{footnote}}


\section{Introduction}

The spatial sign of $x \in \R^p$ is defined as $s(x)=x/|x|$ for $x \neq 0$ and $s(0)=0$, where $|\cdot|$ denotes the Euclidean norm in $\R^p$. For a $p$-dimensional random variable $X$ with distribution $F$ and $t \in \R^p$, $p \geq 2$, we call
\begin{align*}
	S(F,t)	= E\left(s(X-t)s(X-t)^T\right)
\end{align*}
the spatial sign covariance matrix (SSCM) of the distribution $F$ with location $t$. Furthermore, letting $t_n$ be an estimator for $t$ and $\X_n=(X_1,\ldots,X_n)^T$ an $n \times p$ array, where $X_1,\ldots,X_n$ is a random sample from the distribution $F$, we call  
\be \label{eq:S_n}
	S_n	=	S_n(\X_n,t_n)=\frac{1}{n}\sum_{i=1}^n s(X_i-t_n)s(X_i-t_n)^T
\ee
the empirical spatial sign covariance matrix with location $t_n$.
The term \emph{spatial sign covariance matrix} was coined by \citet*{Visuri2000}. \citet*{duerre:vogel:tyler:2014} showed consistency  and asymptotic normality of $S_n$ under mild conditions on $F$ and $t_n$. The estimator $S_n$ has excellent robustness properties. Its influence function is bounded, and its asymptotic breakdown point attains the optimal value of 1/2 \citep{Croux2010}.

We will assume below that $X_1,X_2,\ldots$ follow a continuous elliptical distribution, i.e., $F$ has a Lebesgue density $f$ of the form 
\begin{align*}
	f(x)=\det(V)^{-\frac{1}{2}}g((x-\mu)^TV^{-1}(x-\mu))
\end{align*} 
for a location parameter $\mu \in \R^p$ and a symmetric, positive definite shape matrix $V\in \R^{p\times p}$. We denote the class of continuous elliptical distributions with these parameters by $\Ee_p(\mu,V)$.
The matrix $V$ is called the \emph{shape matrix} of $F$ since it describes the shape of the elliptical contour lines of the density.  The function $g:[0,\infty) \rightarrow [0,\infty)$ is called the \emph{elliptical generator} of $F$. 
The specification of $V$ is unique only up to a multiplicative constant, and therefore $V$ is often 
normalized, e.g., by setting $\det(V)=1$ \citep{Paindaveine2008, Frahm2009}. Since we study scale-free aspects of $F$, where the overall scale is irrelevant, it is more convenient to not fix the scale of the shape. We understand the \emph{shape} of an elliptical distributions as an equivalence class of positive definite matrices being proportional to each other.

There is, up to scale, a one-to-one connection between $S(F,\mu)$ and the parameter $V$: both share the same eigenvectors and the ordering of the respective eigenvalues. 
This makes the spatial sign covariance matrix particularly popular for robust principle component analysis \citep[e.g.][]{marden1999, locantore1999, croux2002, gervini2008}. 
However, the map between the eigenvalues of $V$ and $S(F,\mu)$ is only known explicitly for $p=2$ \citep{Croux2010,vogel2008}. Making use of this result, \citet{duerre:vogel:fried:2015} proposed a robust correlation estimator, called the \emph{spatial sign correlation}. 
Let
\[
	\begin{pmatrix}\hat{s}_{11}&\hat{s}_{12}\\
	\hat{s}_ {21}&\hat{s}_{22}\end{pmatrix}	
	=	S_n(\X_n,t_n)
\]
denote the entries of $S_n(\X_n,t_n$).
Then the generalized correlation coefficient%
\footnote{We call $\rho = v_{12}/\sqrt{v_{11}v_{22}}$ the \emph{generalized correlation coefficient} of the bivariate elliptical distribution $F$. It is defined without any moments assumptions and coincides with the usual product moment correlation if second moments are finite.}
$\rho = v_{12}/\sqrt{v_{11}v_{22}}$ can be estimated by
\be \label{eq:rho_n}
	\hat\rho_n = \frac{c\hat{s}_{12}b}{\sqrt{(\hat{s}_{12}^2+b^2)^2+(\hat{s}_{12}cb)^2}},
\ee
where
\begin{align*}
	b = d-\hat{s}_{11}, \quad
	c = \frac{2d-1}{d(1-d)},\quad
	d = \frac{1}{2}+\sqrt{(\hat{s}_{11}-1/2)^2+\hat{s}_{12}^2}.
\end{align*}
For a derivation of this estimator, see \cite{duerre:vogel:fried:2015}. There it is also shown that the spatial sign correlation $\hat{\rho}_n$ is consistent for $\rho$ under ellipticity and asymptotically normal with asymptotic variance
\be \label{eq:asv.cor1}
	ASV(\hat\rho_n)=(1-\rho^2)^2+ \frac{1}{2}(a+a^{-1})(1-\rho^2)^{3/2}
\ee
where $a=\sqrt{v_{11}/v_{22}}$. 
The asymptotic variance apparently is minimal for $a = 1$, i.e., for equal marginal scales, but it can get arbitrarily large as $a$ approaches $\infty$ or 0. Since our aim is to estimate the correlation coefficient, which is invariant under marginal scale changes, it is therefore reasonable to standardize the data marginally before computing the spatial sign correlation.

Let $\sigma(\cdot)$ denote a univariate \emph{scale measure} or \emph{dispersion measure}, i.e., for any univariate distribution $G$ it satisfies
\be \label{eq:scale.measure}
	\sigma(G^\star_{\alpha,\beta}) = |\alpha|\,\sigma(G)   	
	\qquad \mbox{ for all }  \alpha, \beta \in \R,
\ee  
where $G^\star_{\alpha,\beta}$ is the distribution of $Y^\star = \alpha Y + \beta$ for $Y \sim G$. This may be the standard deviation $\sigma_{SD} = \{E(Y-EY)^2\}^{1/2}$, but since the main purpose of studying spatial sign methods is their robustness, robust measures like the median absolute deviation $\sigma_{\rm MAD} = {\rm median}|Y-{\rm median}(Y)|$ or the $Q_n$ scale measure $\sigma_{Q_n} = q_{1/4}(|Y-Y'|)$ \citep{RousseeuwCroux1993} may be more appropriate. Here, $Y'$ is an independent copy of $Y$, and ${\rm median}(Y)$ denotes the median of the distribution of $Y$ and $q_{1/4}(Y)$ its $1/4$th quantile. 
Let further $\hat\sigma_n = \hat\sigma_n(\Y_n)$ denote the respective \emph{scale estimator}, which is, in principle, the measure $\sigma(\cdot)$ applied to the empirical distribution associated with the univariate sample $\Y_n = (Y_1,\ldots, Y_n)^T$. But in many situations, the empirical version of the scale measure is defined slightly differently due to various reasons, e.g., the empirical standard deviation is usually defined as $\hat\sigma_n(\Y_n) = \{ (n-1)^{-1}\sum_{i=1}^n (Y_i - \bar{Y}_n)^2 \}^{1/2}$ instead of  $\hat\sigma_n(\Y_n) = \{ n^{-1}\sum_{i=1}^n (Y_i - \bar{Y}_n)^2 \}^{1/2}$.

Returning to the general $p$-dimensional set-up, for any specific choice of $\sigma(\cdot)$, let $F_i$ denote the $i$th margin of $F$,
 further $\sigma_i = \sigma(F_i)$ and $\hat\sigma_{i,n} = \hat\sigma_n(\X_n^{(i)})$, where $\X_n^{(i)}$ is the $i$th column of $\X_n$, $1 \le i \le p$. Let
\[
	A = 
	\begin{pmatrix}
		\sigma_1^{-1}	&					&	0		\\
												&	\ddots	&			\\
		0										&					&	\sigma_p^{-1}
	\end{pmatrix}, 
	\qquad
	A_n = 
	\begin{pmatrix}
		\hat{\sigma}_{1,n}^{-1}	&					&	0		\\
												&	\ddots	&			\\
		0										&					&	\hat{\sigma}_{p,n}^{-1}
	\end{pmatrix}. 
\]
Then we define the \emph{two-stage spatial sign  covariance matrix} as
\be \label{eq:tilde.S_n}
	\tilde{S}(\X_n,t_n(\cdot), A_n) = S_n(\X_nA_n, t_n(\cdot))
	=
	\frac{1}{n}\sum_{i=1}^n
	s(A_n X_i-t_n(\X_nA_n))s(A_n X_i-t_n(\X_nA_n))^T,
\ee
and the \emph{two-stage spatial sign correlation} $\hat\rho_{\sigma,n}$ (of the sample $\X_n$ with location $t_n(\cdot)$ and inverse scales $A_n$) as the spatial sign correlation $\hat\rho_n$, cf.~(\ref{eq:rho_n}), being applied to $\tilde{S}_n(\X_n,t_n(\cdot),A_n)$ instead of $S_n(\X_n,t_n)$.
%
%
%
\begin{remark} \label{rem:1} \mbox{ \\ } 
\begin{enumerate}[(I)]
	\item 
There is a subtle but important difference in the role that $t_n$ plays in (\ref{eq:S_n}) and in (\ref{eq:tilde.S_n}).
When defining $S_n(\X_n,t_n)$, the location $t_n$ may generally be any random vector, which may or may not bear a connection to the sample $\X_n$. But usually, we take it to be an estimator computed from the data, i.e., it is a function of $\X_n$. Whenever we want to invoke this latter meaning, we write $t_n(\cdot)$ instead of $t_n$, particularly so in the definition of $\tilde{S}_n(\X_n,t_n(\cdot),A_n)$. Here it is essential that $t_n(\cdot)$ is applied to the transformed data $\X_n A_n$. This will become important at a later point when we consider different location estimates for the transformed data, cf.\ e.g.\ Condition C\ref{c5} of Theorem \ref{th:1} below.

\item \label{rem:1.num:1}
In the definition of the two-stage spatial sign covariance matrix $\tilde{S}_n(\X_n,t_n(\cdot),A_n)$, the data  are \emph{first} standardized marginally, and \emph{then} the location is estimated from the transformed data. For all marginally equivariant location estimators -- and this is vast majority -- the order of these two-steps is irrelevant.
We call a multivariate location estimator $t_n$  \emph{marginally equivariant} if it satisfies $t_n(\X_n A + b) = A t_n(\X_n) + b$ for any $p\times p$ diagonal matrix $A$ 
and $b \in \R^p$. All location estimators being composed of univariate, affine equivariant location estimators are marginally equivariant, but so are also 
all multivariate, affine equivariant location estimators, including elliptical maximum likelihood estimators, $M$-estimators \citep{Maronna1976, Tyler1987}, $S$-estimators \citep{Davies1987}, or constrained $M$-estimators \citep{Kent1996}.
However, there is one prominent example which lacks this property: the spatial median \citep[e.g.][Section 6.2]{oja:2010}. We want to include this estimator since, due to its conceptual similarity to the SSCM, it may be regarded as a default choice for $t_n$. The spatial median has a variety of good properties such as uniqueness and computational and statistical efficiency, see e.g.\ \citet{magyar:2011} and the references therein. Likewise to the spatial sign covariance matrix, the spatial median is inefficient at strongly shaped distributions. Thus, when using the spatial median as location estimate, it is therefore, from a conceptual point of view, reasonable to compute it from the marginally standardized data. This is the reason for choosing the order of steps as we do here: first standardization, then location estimation. However,  in practical situations, the difference to the estimator obtained when reversing the order of these two steps tends to be rather small -- also in case of the spatial median.

\item
Finally we would like stress that we deliberately avoid any reference to the covariance matrix of $F$. Our whole discussion of scale and correlation is completely moment-free.
We understand \emph{correlation} generally as \emph{monotone dependence}, with the moment-based Pearson correlation coefficient being one, and with no doubt the most popular, way of mathematically quantifying this notion. Our main focus here is on estimating the generalized correlation coefficient $\rho$ within the semiparametric model of elliptical distributions, but the concept of spatial sign correlation can also be employed for defining a general, moment-free measure of correlation. Requiring no moment assumptions is one major strength of spatial sign methods.
\end{enumerate}
\end{remark}
Following the introduction, the article has two further sections: Section 2 \emph{Asymptotic results} and Section 3 \emph{Simulations}.
The main result of the paper (Theorem \ref{th:2}) states that, at elliptical distributions, the asymptotic variance of the two-stage spatial sign correlation is
\[
	ASV(\hat\rho_{\sigma,n}) = (1-\rho^2)^2+(1-\rho^2)^{3/2},
\]
which is shown by establishing the asymptotic equivalence of $\hat\rho_{\sigma,n}$ to the spatial sign correlation at distributions with equal marginal scale. This was conjectured by \citet{duerre:vogel:fried:2015}, who compare the corresponding influence functions. Towards this end, we investigate the asymptotics of the two-stage spatial sign covariance matrix (Theorem \ref{th:1}).
With the asymptotic distribution of $\hat\rho_{\sigma,n}$ taking on a rather simple form, only depending on $\rho$, one can derive a variance-stabilizing transformation analogous to Fisher's z-transform. This is the content of Corollary \ref{cor:z-trafo}.
In Section 3, we numerically compare confidence intervals for $\rho$ based on the moment correlation and the spatial sign correlation, both with and without variance-stabilizing transformation. All proofs are deferred to the Appendix.


\section{Asymptotic results}

The first result concerns the asymptotic difference between $S_n(\X_n A_n,t_n)$, the sample two-stage SSCM with estimated location and scales, and $S_n(\X_nA, A t)$, the sample two-stage SSCM with known location and scales. We use the notation $X^{(j)}$ to denote the $j$th component of the $p$-dimensional random vector $X$, $j = 1, \ldots, p$, likewise for other vectors.
\begin{theorem}\label{th:1}
Let $t \in \R^p$ and $X$ 
be a $p$-variate random vector with continuous distribution $F$ satisfying
\begin{enumerate}[(C1)]
\item \label{c1} $E|X - t|^{-3/2} < \infty$, 
\item \label{c2} $E\left\{ \frac{X-t}{|X-t|^2} \right\} =0$ \ and \
			$E\left\{ \frac{ (X-t)^{(i)} (X-t)^{(j)} (X-t)^{(k)} }{|X-t|^4 } \right\} =0$ \ for \  $i,j,k=1,\ldots,p$.
\end{enumerate}
Let further $A$ be a $p\times p$ diagonal matrix with positive diagonal entries $a_1, \ldots, a_p$, and $A_n$ a series of random $p\times p$ diagonal matrices satisfying
\begin{enumerate}[(C1)]
\setcounter{enumi}{2} 
\item \label{c3}
$\sqrt{n}(A_n - A) \cid Z = \diag(Z_1,\dots,Z_p)$
\end{enumerate}
for some random diagonal matrix $Z$. Finally, let $\X_n = (X_1, \ldots, X_n)^T$ be an iid sample drawn from $F$ and $t_n(\cdot)$ a series of $p$-variate estimators satisfying
\begin{enumerate}[(C1)]
\setcounter{enumi}{3} 
\item \label{c4} $\sqrt{n} \{t_n(\X_n) - t\} = O_P(1)$,
\item \label{c5}  $\sqrt{n} \{t_n(\X_n A_n) - A_n t_n(\X_n)\} = O_P(1)$.
\end{enumerate}
Then \ $\sqrt{n} \{ S_n(\X_n A_n, t_n(\cdot)) - S_n(\X_n A, A t) \} \cid \Xi_p$ \ as $n \to \infty$ with
\be \label{eq:xi} 
	\Xi_p =  A^{-1} Z S(F_0, 0) + S(F_0, 0) A^{-1} Z - 2 \sum_{j=1}^p  (Z_j/a_j) \Gamma_j, 
\ee
where $F_0$ is the distribution of $X_0 = A(X-t)$ and
\[
	\Gamma_j = E \left[ (X_0^{(j)})^2  \frac{X_0 X_0^T}{\{X_0^T X_0\}^2} \right].
\]
\end{theorem}

Theorem \ref{th:1} apparently has a long list of technical conditions. They are due to the fact that it is formulated under very broad conditions. We do not assume any specific model for the distribution $F$. Also, the location estimator $t_n(\cdot)$, the scale estimator $A_n$ and even the location $t$ are unspecified. The above conditions are indeed a set of easy-to-verify regularity conditions, which are met in all relevant situations, and many of which may be further relaxed for the price of more involved technical derivations. We will review them one by one below.
\begin{description}
\item[Condition (C\ref{c1})]
requires the probability mass of $F$ to be not too strongly concentrated around $t$. For instance, if $F$ possesses a Lebesgue density $f$, it is sufficient (but not necessary) that $f$ is bounded at $t$. This condition also appears in Theorems 2 and 3 of \citet{duerre:vogel:tyler:2014} and is, loosely speaking, due to the discontinuity of the spatial sign function at the origin. 
\item[Condition (C\ref{c2})]
is indeed a somewhat restrictive condition as it basically imposes component\-wise symmetry of $F$ around $t$. It is, however, a mere convenience assumption, it can be dropped in favor of an additional term in (\ref{eq:xi}) and a slightly stronger formulation of the other conditions (basically joint convergence of $S_n$, $t_n$ and $A_n$). The proof of the more general version runs analogously, with the main difference that \citet[][Theorem 3]{duerre:vogel:tyler:2014} instead of \citet[][Theorem 2]{duerre:vogel:tyler:2014} would be used. However, our central result, Theorem \ref{cor:1} below, concerns elliptical distributions, for which (C\ref{c2}) is fulfilled. We therefore consider it appropriate to include this symmetry condition here for the sake of simpler conditions and a clearer exposition.
\item[Condition (C\ref{c3})]
is satisfied, e.g., if $A_n^{-2}$ is taken to be the diagonal of some $p \times p$ scatter matrix estimator for which asymptotic normality has been shown. But also if $A_n^{-1}$ is composed of univariate scale estimators (the default case here due to computational reasonability), it is usually true. Specifically, if the univariate scale estimator $\hat\sigma_{j,n}$ allows a linearization, i.e., 
\be \label{eq:linearization}
	\hat{\sigma}_{j,n} = \frac{1}{n}\sum_{i=1}^n f_j (X^{(j)}_i) + o_p(n^{-1/2}),
	\qquad j =1,\ldots,p,
\ee
with $E \{f_j(X^{(j)})^2 \}< \infty$, then $\sqrt{n}\{(\hat\sigma_{1,n},\ldots,\hat\sigma_{p,n}) - (\sigma_1,\ldots,\sigma_p)\}^T
= \sqrt{n}\diag(A_n^{-1}-A^{-1})$ converges to a multivariate normal distribution, and then so does $\sqrt{n}(A_n-A)$. Note that, since $A$ and $A_n$ are diagonal matrices, $\sqrt{n}(A_n^{-1}-A^{-1}) \cid \tilde{Z}$ implies $\sqrt{n}(A_n - A) = A A_n \sqrt{n} (A^{-1} - A_n^{-1}) \cid - A^2 \tilde{Z}$, and hence $Z = -A^2\tilde{Z}$ in distribution.

All estimators of practical relevance allow a linearization (\ref{eq:linearization}). For instance, for quantile-based estimators, such as the MAD, this linearization is provided by the Bahadur representation \citep[]{Bahadur1966, kiefer:1967, ghosh:1971, sen:1968}. In the case of $U$-statistics, such as Gini's mean difference, it is given by the Hoeffding decomposition \citep{hoeffding:1948}, and in the case of $U$-quantiles, such as the $Q_n$ scale estimator \citep{RousseeuwCroux1993}, by a combination of the two \citep{serfling:1984,Wendler2011}.

\item[Condition (C\ref{c4}):]
This is a minimal standard assumption.
\item[Condition (C\ref{c5})]
is trivially fulfilled for any marginally equivariant location estimator, see Remark \ref{rem:1} (\ref{rem:1.num:1}).
Primarily, this condition is necessary because we want to include the spatial median as potential location estimator, and, for efficiency reasons, propose to standardize the data prior to computing its spatial median (instead of scaling the spatial median along with the data).  Under (C\ref{c3}), the spatial median satisfies (C\ref{c5}) at elliptical distributions \citep{nevalainen:larocque:oja:2007}. 
\end{description}
Finally, the continuity of $F$ also is a mere convenience assumption, which prohibits that several data points coincide with each other, and thus ensures that $t_n$ coincides with at most one observation. Alternative assumptions are discussed also in \citet{duerre:vogel:tyler:2014}.

In case of $F$ being an elliptical distribution and $t$ its symmetry center, explicit expressions for $S(F,t)$ appear to be known only for $p = 2$. In this case, $\Xi_p$ in (\ref{eq:xi}) considerably simplifies.

\begin{corollary} \label{cor:1}
Let $p =2$ and $X \sim  F \in \Ee_2(t,V)$. Let $A = \diag(a_1,a_2)$ be a $2\times 2$ diagonal matrix with positive diagonal entries such that $V_0 = A V A$ has equal diagonal entries. Let further $Z = \diag(Z_1,Z_2)$ be a random $2 \times 2$ diagonal matrix. Then 
$\Xi_2$ from Theorem \ref{th:1} is
\[
	\Xi_2 = 
	\begin{pmatrix}
	Z_1/a_1 - Z_2/a_2 	&		0								\\
	0										&   Z_2/a_2 - Z_1/a_1 \\
	\end{pmatrix}
	\zeta,
\]
where $\zeta =(1-\sqrt{1-\rho^2})/(2\rho^2)$ if $\rho \neq 0$ and $\zeta = 1/4$ if $\rho=0$, 
and $\rho = v_{12} (v_{11} v_{22})^{-1/2}$.
\end{corollary}
An important implication of Corollary \ref{cor:1} is that, at elliptical population distributions, the asymptotic distribution of the off-diagonal element of the two-dimensional two-stage SSCM is the same as that of the off-diagonal element of the ordinary SSCM at the corresponding distribution with equal marginal scales. Building on this observation, we can derive the asymptotic distribution of the two-stage spatial sign correlation by means of a generalized version of the delta method.
\begin{theorem} \label{th:2}
Let $p = 2$ and $X \sim F \in \Ee_2(t,V)$ satisfy Condition C\ref{c1} of Theorem \ref{th:1}. Let $\X_n$, $A$, $A_n$ and $t_n(\cdot)$ be as in Theorem \ref{th:1}, satisfying Conditions C\ref{c3}, C\ref{c4} and C\ref{c5}, with the further property that $V_0 = A V A$ has equal diagonal entries. Then
\be \label{eq:asv.cor2}
	\sqrt{n}(\hat{\rho}_{\sigma,n}-\rho)\  \cid\   N\big(0,(1-\rho^2)^2+(1-\rho^2)^{3/2}\big).
\ee
\end{theorem}
We have the following remarks about Theorem \ref{th:2}.
\begin{remark} \mbox{ \\ } \label{rem:2}
\begin{enumerate}[(I)]
\item
Comparing (\ref{eq:asv.cor2}) to (\ref{eq:asv.cor1}), we find that, at any elliptical distribution, the spatial sign correlation with the margins being standardized beforehand by the \emph{true} scales and the spatial sign correlation with the margins being standardized by \emph{estimated} scales have the same asymptotic efficiency. In fact, we show in the Appendix that they are asymptotically equivalent. In other words, the loss for not knowing the scale is nil asymptotically, and this is true regardless of the scale estimator used. Any scale function $\sigma(\cdot)$ satisfying (\ref{eq:scale.measure}) yields that that $X_0 = A X$ has equal marginal scales if $X$ is elliptical. Also, the finite-sample variances of the spatial sign correlation with known and estimated scales hardly differ, as the simulations in \citet{duerre:vogel:fried:2015} indicate.
\item
At elliptical distributions with finite fourth moments, the asymptotic variance of the product moment correlation is $(1+\kappa/3)(1-\rho^2)^2$, where $\kappa$ is the marginal excess kurtosis. Thus under normality, where $\kappa = 0$, the additional term  
$(1-\rho^2)^{3/2}$ may be viewed as the price to pay efficiency-wise for the gain in robustness when using the spatial sign correlation instead of the moment correlation.
\item
In case of a two-dimensional elliptical distribution, Condition (C\ref{c1}) is fulfilled if $g(z) = O(z^{-1/4+\delta})$ as $z \to 0$ for some $\delta > 0$.
\end{enumerate}
\end{remark}
The asymptotic distribution of $\hat\rho_{\sigma,n}$ only depends on $\rho$, but not on the elliptical generator $g$ or any other characteristic of the population distribution. Therefore the two-stage spatial sign correlation is very well suited for nonparametric and robust correlation testing. Likewise to Fisher's $z$-transformation for the moment correlation under normality \citep{fisher:1921, hotelling:1953}, one can find a variance-stabilizing transformation for the spatial sign correlation under ellipticity.
\begin{corollary} \label{cor:z-trafo}
Under the conditions of Theorem \ref{th:2}, we have $\sqrt{n} \{ h(\hat{\rho}_{\sigma,n})-h(\rho) \} \cid N(0,1)$ 
with
\begin{align*}
	h(x)=s(x)\left(\frac{1}{\sqrt{2}}\arcsin\left(\frac{3(1-\sqrt{1-x^2})-2}{\sqrt{1-x^2}+1}\right)+\frac{\pi}{2^{3/2}}\right),
\end{align*}
where $s(\cdot)$ denotes the (in this case univariate) sign function. 
\end{corollary}
\begin{figure}\label{fig:fisher}
\begin{center}
\includegraphics[width=0.42\textwidth]{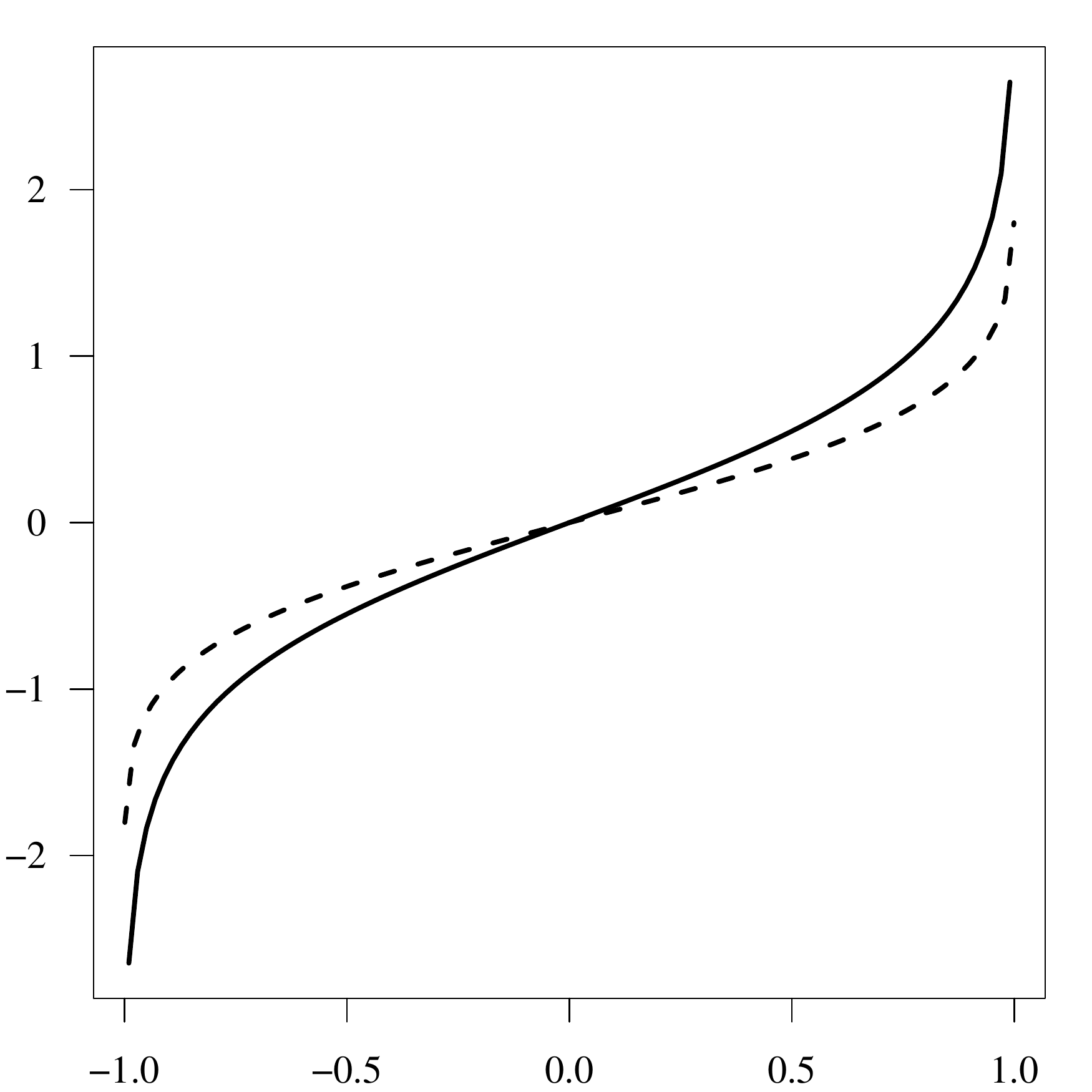}
\hfill
\includegraphics[width=0.42\textwidth]{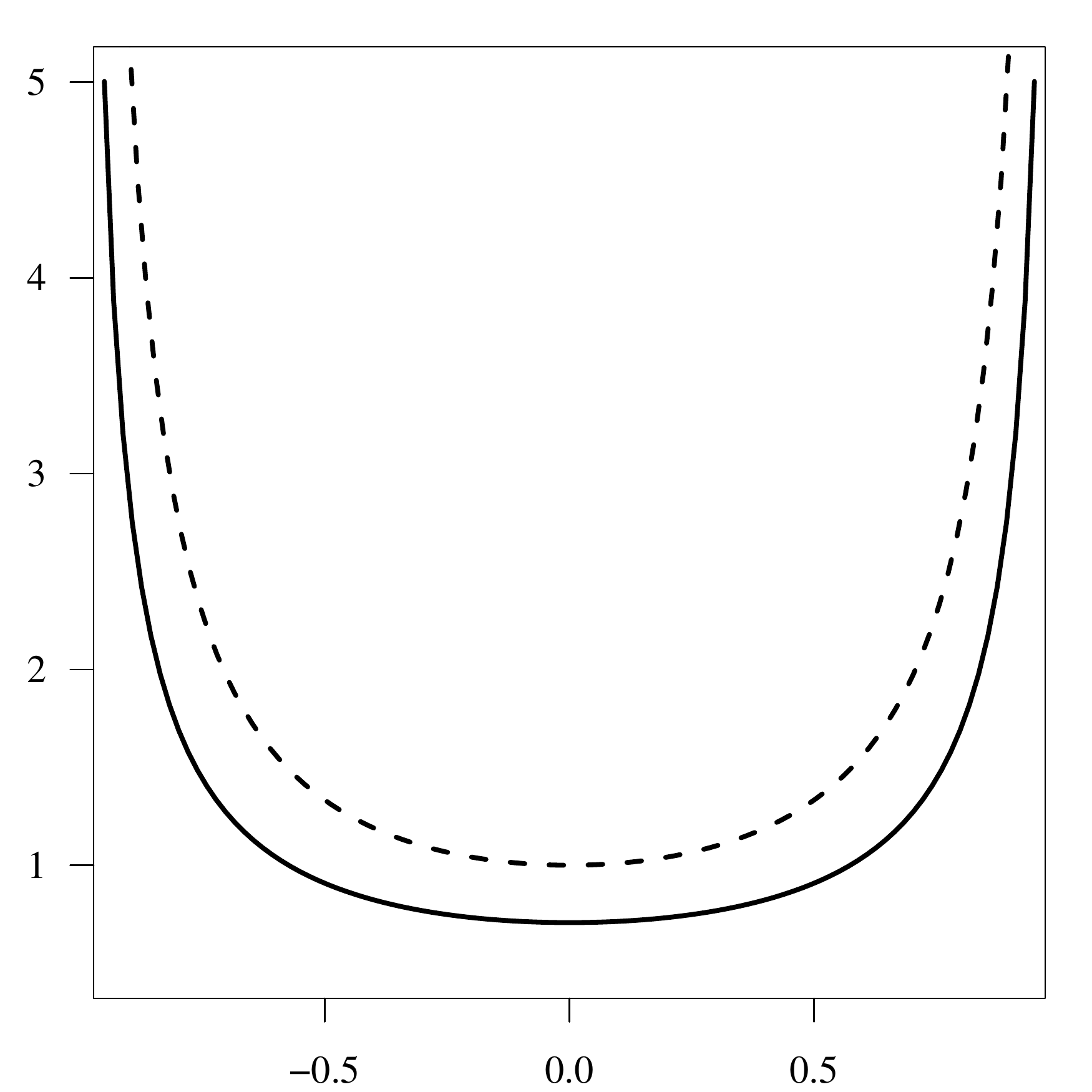}
\caption{Variance-stabilizing transformations (left) and their derivates (right) for the spatial sign correlation (solid) and the Pearson moment correlation, i.e., Fisher's $z$-transform (dashed). }
\end{center}
\end{figure}
As can be seen in Figure~\ref{fig:fisher}, the transformation $h$ is similar to Fisher's $z$-transform $x \mapsto \log\{(1+x)/(1-x)\}/2$. There are two main differences: first, $h$ is flatter, with a smaller derivative throughout, reflecting the larger asymptotic variance of the spatial sign correlation under normality, and second, $h$ is bounded, attaining only values between $-\pi/\sqrt{2}$ and $\pi/\sqrt{2}$. 
To construct confidence intervals, its inverse function $h^{-1}:[-\pi/\sqrt{2},\pi/\sqrt{2}]\rightarrow [-1,1]$ is also of interest. It is given by
\begin{align*}
	h^{-1}(y) = s(y)\frac{2^{3/2}\sqrt{1-\cos(\sqrt{2}y)}}{3-\cos(\sqrt{2}y)}.
\end{align*}  
Based on Corollary \ref{cor:z-trafo}, one can derive asymptotic level-$\alpha$-tests for the generalized correlation coefficient $\rho$ of a bivariate elliptical distribution, which are robust and very accurate also in small samples, as the results of Section \ref{sec:sim} below indicate.
For instance, a two-sided one-sample test for $\rho$ based on $\hat\rho_{\sigma,n}$ would reject the null hypothesis $\rho = \rho_0$ at the significance level $\alpha$ if the test statistic
\[
	T_{1,n} = n \{ h(\hat{\rho}_{\sigma,n}) - h(\rho_0)\}^2
\]
exceeds $\chi_{1;1-\alpha}^2$, i.e., the $1-\alpha$ quantile of the $\chi^2$ distribution with one degree of freedom. Likewise, for two samples of sizes $n_1$ and $n_2$ and generalized correlation coefficients $\rho^{(1)}$ and $\rho^{(2)}$, respectively, the null hypothesis $\rho^{(1)} = \rho^{(2)}$ is rejected if 
\[
	T_{2,n} = \frac{n_1 n_2}{n_1+n_2} \{ h(\hat{\rho}_{\sigma,n}^{(1)}) - h(\hat{\rho}_{\sigma,n}^{(2)})\}^2
\]
is larger than $\chi_{1;1-\alpha}^2$, where $\hat{\rho}_{\sigma,n}^{(i)}$, $i = 1,2$, denote the two-stage spatial sign correlations computed from the two samples. Similarly, one can construct one-sided and $k$-sample tests.


\section{Simulations}
\label{sec:sim}

We want to numerically investigate the usefulness of the asymptotics in finite samples. 
We compute 95\% confidence intervals based on the spatial sign correlation with and without the transformation $h$, denoted in the tables below by sscor-$h$ and sscor, respectively. 
The simulations are done with the statistical software R \citep{R}. 
We sample from bivariate elliptical distributions using the package mvtnorm \citep{mvtnorm}.
The central location is computed by the spatial median from the package pcaPP \citep{pcaPP}, and the scales are estimated by the $Q_n$ implemented in the package robustbase \citep{robustbase}.

Pearson's moment correlation with  and without Fisher's $z$-transform (denoted by cor-$z$ and cor, respectively) serves as a benchmark. Under ellipticity, the asymptotic variance of the moment correlation additionally depends on the kurtosis $\kappa$. We estimate the latter by the following multivariate kurtosis estimator 
\begin{align*}
	\hat\kappa_n 
	= \frac{3}{p(p+2)}\frac{1}{n}\sum_{i=1}^n \{ (X_i-\bar{X}_n)^T \hat\Sigma_n^{-1}(X_i-\bar{X}_n)\}^2  \ - \ 3,
\end{align*}
where $\bar{X}_n$ denotes the sample mean and $\hat\Sigma_n$ the sample covariance matrix \citep[e.g.][p.~103]{Anderson2003}. Alternatively, one may estimate the kurtosis by averaging the componentwise marginal sample kurtoses, as it is done, e.g., in \citet{Vogel2011}.

In Table \ref{tab:1}, covering frequencies of the generalized correlation coefficient $\rho$ by the various confidence intervals are given based on 10,000 repetitions for each parameter setting. We consider the normal distribution and the $t$-distribution with 5 and 3 degrees of freedom, true correlations of $\rho = 0$ and $\rho = 0.5$ and six different sample sizes ranging from $n = 10$ to $n = 10,000$.
We see that the sscor-$h$ confidence intervals, i.e., the spatial-sign-based with transformation $h$, are almost exact in all cases considered, already for $n = 10$. The spatial-sign-based confidence intervals without transformation reach a comparable accuracy only for $n = 50$, and confidence intervals based the Pearson correlation (with and without $z$-transformation) no sooner than $n = 100$ at normality and $n = 500$ at the $t_5$ distribution.
Table~\ref{tab:2} reports the corresponding average lengths of the confidence intervals multiplied by $\sqrt{n}$. Comparing these average lengths for the Pearson correlation and spatial sign correlation, we rediscover roughly the square root of the ratio of the asymptotic variances, e.g., for the normal distribution at $\rho = 0$, we have $5.54/3.92 = 1.413 \approx \sqrt{2}$.
At normality, the confidence intervals based on the Pearson correlation (the maximum likelihood estimator for $\rho$ in this case) are shorter, whereas the sscor confidence intervals are shorter at the $t_5$ distribution -- at least in larger samples, where all confidence intervals have the same 95\% covering probability. Thus, in a heavy-tailed setting like the $t_5$ distribution, the spatial sign based confidence intervals are superior -- in terms of covering accuracy as well as length.
Further, we observe that the strict asymptotic distribution-freeness of the spatial sign correlation practically also extends to the finite-sample case. In both tables, the results for the spatial sign correlation are essentially the same for the three different elliptical distributions. In contrast, the Pearson correlation shows a considerably worse finite-sample behavior at the $t_5$  than at the normal distribution.

The fourth moments of the $t_3$ distribution are not finite, i.e., the kurtosis does not exist, and the moment correlation is not $\sqrt{n}$-consistent when sampling from a $t_3$ distribution. Hence the usual construction of the moment correlation based confidence intervals has no mathematical justification. However, the bottom parts of Tables \ref{tab:1} and \ref{tab:2} indicate that, when ignoring this fact, Pearson's moment correlation nevertheless provides somewhat useful, approximate confidence intervals. While for small $n$ the moment-correlation-based confidence intervals are short but have a too low coverage probability, they reach 95\% in large samples, but are in comparison to, e.g., the sscor based confidence intervals very large.
This somewhat unexpected observation is not completely surprising, since the length of the confidence intervals is largely determined by the sample kurtosis. The slower convergence of the sample moment correlation to $\rho$, and the exploding behavior of the sample kurtosis are opposing effects, which appear to basically cancel each other.

Altogether the spatial correlation with variance stabilizing transformation $h$ yields very reliable confidence bands, which are accurate also in very small samples.

\begin{table}[t] 
\begin{center}
\begin{tabular}{r||cc|cc||cc|cc}
\hline
$\rho$ & \multicolumn{4}{c||}{0} & \multicolumn{4}{c}{0.5}\\
$n$ & sscor & sscor-$h$ &  cor & cor-$z$ & sscor & sscor-$h$ & cor & cor-$z$\\  
\hline
\hline
\multicolumn{9}{l}{ normal distribution  } \\
\hline
%
		10 &   86 &   94 &   77 &   83 &  87 &   93 &   78 &   83 \\
		20 &   90 &   94 &   86 &   89 &  91 &   95 &   87 &   90 \\
		50 &   93 &   95 &   92 &   93 &  93 &   95 &   92 &   93\\
	 100 &   93 &   95 &   93 &   94 &  94 &   95 &   93 &   94 \\
   500 &   95 &   95 &   95 &   95 &  95 &   95 &   95 &   95\\
 10000 &   95 &   95 &   95 &   95 &  95 &   95 &   95 &   95 \\
\hline
\hline
\multicolumn{9}{l}{ $t_5$ distribution  } \\
\hline
		10 &   85 &   94 &   70 &   76 &    87 &   93 &   71 &   77 \\
		20 &   90 &   95 &   81 &   85 &    90 &   95 &   80 &   85 \\
		50 &   93 &   95 &   88 &   90 &    93 &   95 &   88 &   90 \\
	 100 &   94 &   95 &   91 &   92 &    94 &   95 &   91 &   92 \\
	 500 &   95 &   95 &   94 &   94 &    95 &   95 &   94 &   94 \\
 10000 &   95 &   95 &   95 &   95 &    95 &   95 &   95 &   95 \\
\hline
\hline
\multicolumn{9}{l}{ $t_3$ distribution  } \\
\hline
		10 &  85 &  94 &  64 &  71 & 		  87 &  93 &  66 &  72 \\
		20 &  90 &  94 &  74 &  79 &  	  90 &  94 &  76 &  81 \\
		50 &  93 &  95 &  82 &  86 &  	  93 &  95 &  82 &  85 \\
	 100 &  94 &  95 &  86 &  88 &  	  94 &  95 &  86 &  88 \\
	 500 &  95 &  95 &  90 &  91 &  	  94 &  95 &  90 &  92 \\
 10000 &  95 &  95 &  94 &  95 &    	95 &  95 &  94 &  94 \\
\hline
\end{tabular}
\caption{  \label{tab:1} Empirical covering probabilities (\%) of asymptotic 95\% confidence intervals based on the spatial sign correlation (sscor) and the moment correlation (cor) with and without variance-stabilizing transformation for bivariate normal and $t$-distributions with 3 and 5 degrees of freedom, $\rho =0$ and $\rho = 0.5$, and varying sample sizes $n$; 10,000 repetitions.}
\end{center}
\end{table}


\begin{table}[t]
\begin{center}
\begin{tabular}{r||cc|cc||cc|cc}
\hline
$\rho$ & \multicolumn{4}{c||}{0} & \multicolumn{4}{c}{0.5}\\
$n$ & sscor & sscor-$h$ & cor & cor-$z$ & sscor & sscor-$h$ & cor & cor-$z$\\
\hline
\hline
\multicolumn{9}{l}{ normal distribution  } \\
\hline
		10 &  4.75 &  4.11 &  2.83 &  2.67 &    4.08 &  3.69 &  2.23 &  2.17 \\
		20 &  5.11 &  4.68 &  3.35 &  3.20 &   	4.20 &  3.99 &  2.57 &  2.53 \\
		50 &  5.36 &  5.15 &  3.68 &  3.60 &    4.26 &  4.18 &  2.79 &  2.77 \\
	 100 &  5.45 &  5.33 &  3.80 &  3.75 &    4.30 &  4.26 &  2.87 &  2.86 \\
	 500 &  5.52 &  5.50 &  3.89 &  3.88 &   	4.31 &  4.30 &  2.92 &  2.92 \\
 10000 &  5.54 &  5.54 &  3.92 &  3.92 &    4.32 &  4.32 &  2.94 &  2.94 \\
\hline
\hline
\multicolumn{9}{l}{ $t_5$ distribution  } \\
\hline
		10  &  4.75  &  4.12  &  2.84  &  2.68  &  	4.08  &  3.69  &  2.28  &  2.21  \\
		20  &  5.11  &  4.68  &  3.63  &  3.45  &  	4.22  &  4.01  &  2.83  &  2.77  \\
		50  &  5.36  &  5.15  &  4.44  &  4.30  &   4.28  &  4.20  &  3.40  &  3.36  \\
	 100  &  5.45  &  5.34  &  4.92  &  4.82  &  	4.30  &  4.26  &  3.74  &  3.71  \\
	 500  &  5.52  &  5.50  &  5.71  &  5.67  &   4.31  &  4.31  &  4.29  &  4.28  \\
 10000  &  5.54  &  5.54  &  6.38  &  6.38  &   4.32  &  4.31  &  4.79  &  4.79  \\
\hline
\hline
\multicolumn{9}{l}{ $t_3$ distribution  } \\
\hline
		10 &  4.73 &  4.10 &  2.81 	&  2.65 	&    	4.09 &  3.70 &  2.27 	&  2.21 \\
		20 &  5.11 &  4.68 &  3.77 	&  3.58 	&    	4.22 &  4.01 &  2.99 	&  2.92 \\
		50 &  5.36 &  5.15 &  5.08 	&  4.87 	&   	4.28 &  4.20 &  3.94 	&  3.87 \\
	 100 &  5.45 &  5.34 &  6.15 	&  5.95 	&    	4.30 &  4.26 &  4.72 	&  4.66 \\
	 500 &  5.52 &  5.50 &  9.12 	&  8.96 	&    	4.31 &  4.31 &  6.90 	&  6.85 \\
 10000 &  5.54 &  5.54 &  17.57 &  17.46 	&   	4.32 &  4.32 &  13.02 &  12.99 \\
\hline
\end{tabular}
\caption{  \label{tab:2} Average lengths of 95\% confidence intervals based on the spatial sign correlation (sscor) and the moment correlation (cor) with and without variance-stabilizing transformation for bivariate normal and $t$-distributions with 3 and 5 degrees of freedom, $\rho =0$ and $\rho = 0.5$, and varying sample sizes $n$; 10,000 repetitions.}
\end{center}
\end{table}


\section{Conclusion}

The spatial sign correlation, as introduced in \citet{duerre:vogel:fried:2015}, cf.~(\ref{eq:rho_n}), is a robust correlation estimator which has a variety of nice properties. It is fast to compute, it is distribution-free within the elliptical model, its efficiency is comparable to other estimators offering a similar degree of robustness, and the explicit form of the asymptotic variance facilitates inferential procedures. In this article we have addressed its main drawback: the inefficiency under strongly shaped models, i.e., where the eigenvalues of the shape matrix strongly differ. The shapedness due to different marginal scales may be eliminated by a componentwise standardization before computing the the spatial sign correlation. We have shown that the resulting two-step estimator has the same asymptotic distribution as the spatial sign correlation applied to a sample from a model with equal marginal scales.
An important consequence is that the parameter $a$, cf.~(\ref{eq:asv.cor1}), i.e., the ratio of the marginal scales, drops from the expression for the asymptotic variance. The only parameter left is the generalized correlation coefficient $\rho$ itself. This allows to devise a variance-stabilizing transformation  similar to Fisher's $z$-transformation, which, contrary to Fisher's transform, is valid for all elliptical distributions.
The prior standardization makes the spatial sign correlation really a practical estimator. 


\section*{Acknowledgments} This research was supported in part by the Collaborative Research Grant 823 of the German Research Foundation. Credit must be given to the anonymous referees of the article \emph{Spatial sign correlation} 
(J.\ Multivariate Anal.\ 135, pages 89--105, 2015), who independently of each other suggested to further explore the properties of two-stage spatial sign correlation, which initiated the research presented in this article.


\appendix
\section{Proofs}
In the proof of Theorem \ref{th:1}, we make use of the following lemma, which states that the empirical versions of $\Gamma_j$ with and without location estimation are asymptotically equivalent.
\begin{applem} \label{applem:1}
Let $t \in \R^p$ and $X$ be a $p$-variate random vector with distribution $F$ satisfying
\begin{enumerate}[(I)]
\item $E|X-t|^{-2/3} < \infty$. \label{num:applem1}
\end{enumerate}
Let further $\X_n = (X_1, \ldots, X_n)^T$ be an iid sample drawn from $F$ and $t_n$ a series of $p$-variate random vectors satisfying
\begin{enumerate}[(I)]
\setcounter{enumi}{1} 
\item \label{num:applem2} $\sqrt{n} (t_n - t) = O_P(1)$.
\end{enumerate}
Finally, let $A$ be a diagonal $p\times p$ matrix with positive diagonal entries. Then, for all $1 \le j \le p$,
\begin{align*}
	\hspace{-2.5em}
	\frac{1}{n}\sum_{i=1}^n 
	\left[ a_j^2 (X_i^{(j)}\!\!-\!t_n^{(j)})^2 \frac{A(X_i\!-\!t_n)(X_i\!-\!t_n)^TA}{\{(X_i\!-\!t_n)^T A^2 (X_i\!-\!t_n)\}^2}\right] -
	\frac{1}{n}\sum_{i=1}^n 
	\left[a_j^2 (X_i^{(j)}\!\!-\!t^{(j)})^2 \frac{A(X_i\!-\!t)(X_i\!-\!t)^TA}{\{(X_i\!-\!t)^T A^2 (X_i\!-\!t)\}^2} \right]
\end{align*}
converges to zero in probability as $n \to \infty$.
\end{applem}
\begin{proof}
To shorten notation and without loss of generality we will assume that $t=0$ and $A = I_p$. We will show componentwise convergence, i.e.
\be \label{eq:componentwise}
	\frac{1}{n}\sum_{i=1}^n \left[ 
			\frac{(X_i^{(j)}-t_n^{(j)})^2(X_i-t_n)^{(k)}(X_i-t_n)^{(l)}}{|X_i-t_n|^4}
		- \frac{(X_i^{(j)})^2 X_i^{(k)}X_i^{(l)} }{|X_i|^4}
	\right] \cip 0
\ee
as $n \to \infty$ for all $1 \le j, k, l \le p$.
We use the following random partition of $\R^p$:
\begin{align}\label{eq:partition}
	B_n = \{x\in \R^p|~|x-t_n|\geq \frac{1}{2}|x|\}, \qquad 
	B_n^C=\{x\in \R^p|~|x-t_n|< \frac{1}{2}|x|\}.
\end{align}
and the corresponding random partition of the index set $\{ 1, \ldots, n \}$:
\begin{align*}
	I_n=\{1\leq i \leq n|X_i\in B_n\}, \qquad
	I_n^C=\{1\leq i \leq n|X_i\in B_n^C\}.
\end{align*}
Letting $K_i$ denote the summands in (\ref{eq:componentwise}), we write
\begin{align}\label{eq:decomposition} 
		\frac{1}{n} \sum_{i=1}^n K_i
	=	\frac{1}{n}\sum_{i\in I_n}K_i
	+ \frac{1}{n}\sum_{i\in I_n^C}K_i.
\end{align}
For the second sum on the right-hand side of (\ref{eq:decomposition}) we make use of $|K_i| \le 2$ to obtain
$| n^{-1}\sum_{i \in I_n^C}K_i| \leq 2 n^{-1} \sum_{i=1}^n \mathds{1}_{B_n^C}(X_i)$. The right-hand side of the last inequality is shown to converge to zero in probability under the assumptions of Lemma A\ref{applem:1} as in the proof of Theorem 1 in \citet[][]{duerre:vogel:tyler:2014}. The first sum on the right-hand side of (\ref{eq:decomposition}) is decomposed into
\[
	\frac{1}{n}\sum_{i\in I_n} K_i \ = \ 
	\frac{1}{n}\sum_{i\in I_n}
		\frac{ \{ (X_i^{(j)}-t_n^{(j)})^2-(X_i^{(j)})^2\} (X_i-t_n)^{(k)}(X_i-t_n)^{(l)}|X_i|^4}{|X_i-t_n|^4|X_i|^4}
\]\[
	\quad + \  \frac{1}{n}\sum_{i\in I_n}
		\frac{(X_i^{(j)})^2 t_n^{(k)} (X_i-t_n)^{(l)}|X_i|^4}{|X_i-t_n|^4|X_i|^4}
	\ + \ \frac{1}{n}\sum_{i\in I_n}
	 \frac{(X_i^{(j)})^2 X_i^{(k)} t_n^{(l)} |X_i|^4}{|X_i-t_n|^4|X_i|^4}
\]\[
	\quad + \ \frac{1}{n}\sum_{i\in I_n}\frac{(X_i^{(j)})^2X_i^{(k)}X_i^{(l)}(|X_i|^4-|X_i-t_n|^4)}{|X_i-t_n|^4|X_i|^4}\label{G4}.
\]
Call the four terms from left to right $\calT_1$, $\calT_2$, $\calT_3$, $\calT_4$. Since $X_i \in B_n$ implies $|X_i|  \le 2 |X_i - t_n|$, we have
\[
	|\calT_1| \le \frac{1}{n}\sum_{i\in I_n}
	\left\vert 
		\frac{ t_n^{(j)} (t_n^{(j)} - 2 X_i^{(j)}) (X_i-t_n)^{(k)}(X_i-t_n)^{(l)}|X_i|^4}{|X_i-t_n|^4|X_i|^4}
	\right\vert
\]\[
	\quad \le
	\frac{1}{n}\sum_{i\in I_n}
	\left\vert \frac{ t_n^{(j)} (t_n^{(j)} - X_i^{(j)}) }{|X_i-t_n|^2} \right\vert
	+
		\frac{1}{n}\sum_{i\in I_n}
	\left\vert \frac{ t_n^{(j)}  X_i^{(j)} }{|X_i-t_n|^2} \right\vert
\]\[
	\quad \le
	\frac{2}{n}\sum_{i\in I_n} 	\frac{ | t_n^{(j)}|  }{|X_i|}  
	+
		\frac{4}{n}\sum_{i\in I_n} 	\frac{ | t_n^{(j)}|  }{|X_i|} 
	\ \le \
		\frac{6}{n} \sum_{i=1}^n 	\frac{ | t_n^{(j)}|  }{|X_i|} 
	\ = \ 
	6 \sqrt{n} | t_n^{(j)}| \left\{ \frac{1}{n^{3/2}} \sum_{i=1}^n \frac{1}{|X_i|}\right\} \cip 0,
\]
since the term in $\{\cdot\}$ converges to zero almost surely by Marczinkiewicz's law of large numbers \citep[p.~255]{loeve:1977}.
Convergence to zero of the remaining terms $\calT_2$, $\calT_3$ and $\calT_4$ is shown analogously. The proof of Lemma \ref{applem:1} is complete. 
\end{proof}
Remark: One can see from the last displayed line that, similarly to Theorem 1 of \citet{duerre:vogel:tyler:2014}, the lemma can be proven also under slightly different conditions. For instance, assumption (\ref{num:applem2}) can weakened to $t_n \cip t$ in exchange for the stronger moment condition $E|X-t|^{-1} < \infty$.
%
%
%
%

We are now ready to prove Theorem \ref{th:1}.
\begin{proof}[Proof of Theorem \ref{th:1}]
Let $\tilde{t}_n(\X_n) = A_n^{-1} t_n(\X_n A_n)$%
\footnote{Technically, $\tilde{t}_n$ is a function of $\X_n$ as well as $A_n$. We can understand $\tilde{t}_n(\X_n)$ as a short-hand notation, where the dependence on $A_n$ is simply suppressed, but the notation is also justified in the sense that $A_n$ usually is a function of $\X_n$.}  
and write $\sqrt{n} \{ S_n(\X_n A_n, t_n(\cdot)) - S_n(\X_n A, A t) \}$ as
\[
	\frac{1}{\sqrt{n}} \sum_{i=1}^n \left[
		\frac{ \{A_n X_i - A_n \tilde{t}_n(\X_n)\} \{A_n X_i - A_n \tilde{t}_n(\X_n)\}^T }{ \{A_n X_i - A_n \tilde{t}_n(\X_n)\}^T \{A_n X_i - A_n \tilde{t}_n(\X_n)\} }
		-
		\frac{ \{A X_i - A \tilde{t}_n(\X_n)\} \{A X_i - A \tilde{t}_n(\X_n)\}^T }{ \{A X_i - A \tilde{t}_n(\X_n)\}^T \{A X_i - A \tilde{t}_n(\X_n)\} }
	\right]
\]
\[
	+ \ \frac{1}{\sqrt{n}} \sum_{i=1}^n \left[
		\frac{  \{A X_i - A \tilde{t}_n(\X_n)\} \{A X_i - A \tilde{t}_n(\X_n)\}^T }{ \{A X_i - A \tilde{t}_n(\X_n)\}^T \{A X_i - A \tilde{t}_n(\X_n)\} } 
		-
		\frac{ \{A X_i - A t \} \{A X_i - A t \}^T }{ \{A X_i - A t \}^T \{A X_i - A t \} } 
	\right].
\]
Call the first term $\mathcal{T}_1$ and the second $\mathcal{T}_2$. The convergence of $\mathcal{T}_2$ to zero in probability follows with \citet[][Theorem~2]{duerre:vogel:tyler:2014}. Let $\tilde{X}_i = A X_i$, $\tau_n = A \tilde{t}_n$ and $\tau = A t$. Then Theorem 2 of \citet{duerre:vogel:tyler:2014}
 essentially states that $\sqrt{n} \{ S_n(\tilde{\X}_n, \tau_n) - S_n(\tilde{\X}_n, \tau) \} \cip 0$, where $\tilde\X_n = (\tilde{X}_1, \ldots, \tilde{X}_n)^T$. This is not stated explicitly in the text of the theorem, but this is what is proven. To check that the assumptions are met, note that by Conditions (C\ref{c3}), (C\ref{c4}) and (C\ref{c5}) we have
\[
	\sqrt{n} (\tau_n - \tau) = A \sqrt{n} ( A_n^{-1} t_n(\X_n A_n) - t )
\]
\[
 = \ A \sqrt{n} ( A_n^{-1} t_n(\X_n A_n) - t_n(\X_n) ) \ +\ A \sqrt{n} ( t_n(\X_n) - t ) 
\]
\[
	= \ \sqrt{n}  A A_n^{-1} ( t_n(\X_n A_n) - A_n t_n(\X_n) ) \ +\ A \sqrt{n} ( t_n(\X_n) - t ) \ =\ O_P(1).
\]
The latter is sufficient (along with continuity of $F$), cf.\ the remarks below Theorem~3 in \citet[][]{duerre:vogel:tyler:2014}.
We are thus left to prove $\mathcal{T}_1 \cid \Xi_p$. Let $Y_i = X_i - \tilde{t}_n(\X_n)$, where we suppress the dependence the on $n$ in this short-hand notation. Then $\mathcal{T}_1$ can be further decomposed into
\[
	\mathcal{T}_1 \ = \  
	\frac{1}{\sqrt{n}} \sum_{i=1}^n
	\frac{A_nY_iY_i^TA_n -AY_iY_i^TA}{Y_i^TA^2Y_i}
	\ + \
	\frac{1}{\sqrt{n}} \sum_{i=1}^n\frac{ \{Y_i^T(A^2-A_n^2)Y_i\} \, A_nY_iY_i^TA_n}{Y_i^TA_n^2Y_iY_i^TA^2Y_i}.
\]
We call the terms $\mathcal{T}_{1,a}$ and $\mathcal{T}_{1,b}$, where we have 
\[
	\mathcal{T}_{1,a} \ = \ 
	A_n A^{-1}\left(\frac{1}{n}\sum_{i=1}^n \frac{AY_iY_i^TA}{Y_i^TA^2Y_i}\right)A^{-1}\sqrt{n}(A_n-A)
	+	
	\sqrt{n}(A_n-A)A^{-1}\left(\frac{1}{n}\sum_{i=1}^n\frac{AY_iY_i^TA}{Y_i^TA^2Y_i}\right),
\]
which converges in distribution to $S(F_0,0)A^{-1}Z + Z A^{-1} S(F_0,0)$, since 
\[
	\frac{1}{n}\sum_{i=1}^n\frac{AY_iY_i^TA}{Y_i^TA^2Y_i} \cip S(F_0,0)
\]
by Theorem 1 in \citet{duerre:vogel:tyler:2014}. Writing $\mathcal{T}_{1,b}$ as $\mathcal{T}_{1,b} = \mathcal{L} + \mathcal{R}$ with
\[
	\mathcal{L} = \frac{1}{\sqrt{n}} 	\sum_{i=1}^n \frac{ 2 \{Y_i^T(A-A_n)A\} \, Y_iAY_iY_i^TA} {(Y_i^TA^2Y_i)^2},
\]\[
	\mathcal{R} 
	= \frac{1}{\sqrt{n}} 	\sum_{i=1}^n
	\left\{
			\frac{ \{ Y_i^T(A-A_n)(A+A_n) Y_i \}\, A_nY_iY_i^TA_n} {Y_i^TA_n^2Y_iY_i^TA^2Y_i} 
	-  2\frac{\{ Y_i^T(A-A_n)A Y_i \}\, AY_iY_i^TA}{(Y_i^TA^2Y_i)^2}
	\right\},
\]
we find for $\mathcal{L}$ by using Lemma A\ref{applem:1}
\[
	\mathcal{L} = 
	2 \sum_{j=1}^p \{A^{-1} \sqrt{n}(A-A_n) \}^{(j,j)}
	\frac{1}{n}\sum_{i=1}^n  \{(A Y_i)^{(j)}\}^2 \frac{AY_iY_i^TA} {(Y_i^TA^2Y_i)^2}
	\cid 
	- 2 \sum_{j=1}^p (A^{-1}Z)^{(j,j)}\Gamma_j.
\]
It remains to show that $\mathcal{R}$ vanishes asymptotically. We further decompose $\mathcal{R}$ into
\[
	\frac{1}{\sqrt{n}} 	\sum_{i=1}^n \left[
			\frac{ \{Y_i^T(A\!-\!A_n)(A\!+\!A_n)Y_i\}\, A_nY_iY_i^TA_n}{Y_i^TA_n^2Y_iY_i^TA^2Y_i}
		- \frac{ \{Y_i^T(A\!-\!A_n)(A\!+\!A_n)Y_i\}\, A_nY_iY_i^TA_n}{Y_i^TA^2Y_iY_i^TA^2Y_i}
	\right]
\]\[
	+ \ \frac{1}{\sqrt{n}} 	\sum_{i=1}^n \left[
			\frac{ \{ Y_i^T(A\!-\!A_n)(A\!+\!A_n)Y_i\}\,  A_nY_iY_i^TA_n}{Y_i^TA^2Y_iY_i^TA^2Y_i}
		- \frac{ \{ Y_i^T(A\!-\!A_n)2 A Y_i\}\, A_nY_iY_i^TA_n}{Y_i^TA^2Y_iY_i^TA^2Y_i}
		\right]
\]\[
		+ \	\frac{1}{\sqrt{n}} 	\sum_{i=1}^n 2 \left[
			\frac{\{Y_i^T(A\!-\!A_n)AY_i\}\, A_nY_iY_i^T A_n}{Y_i^TA^2Y_iY_i^TA^2Y_i} 
		- \frac{\{Y_i^T(A\!-\!A_n)AY_i\}\, AY_iY_i^TA_n}{Y_i^TA^2Y_iY_i^TA^2Y_i}
	\right]
\]\[
		+ \	\frac{1}{\sqrt{n}} 	\sum_{i=1}^n 2 \left[
			\frac{ \{ Y_i^T(A\!-\!A_n)AY_i\}\, AY_iY_i^TA_n}{Y_i^TA^2Y_iY_i^TA^2Y_i} 
		- \frac{ \{ Y_i^T(A\!-\!A_n)AY_i\}\, AY_iY_i^TA}{Y_i^TA^2Y_iY_i^TA^2Y_i}
	\right]
\]
and denote the four terms by $\mathcal{S}_1$, $\mathcal{S}_2$, $\mathcal{S}_3$ and $\mathcal{S}_4$, respectively. For $\mathcal{S}_1$ we get
\[
	| \mathcal{S}_1| 
	\le 
	\frac{1}{\sqrt{n}} 	\sum_{i=1}^n
	\left\vert 
		\frac{\{ Y_i^T(A-A_n)(A+A_n)Y_i\}^2 A_nY_iY_i^TA_n}{Y_i^TA_n^2Y_i(Y_i^TA^2Y_i)^2}
	\right\vert
\]\[
	\le 
	\frac{1}{\sqrt{n}} 	\sum_{i=1}^n
	\left\{
		\frac{ Y_i^T(A-A_n)(A+A_n)Y_i}{Y_i^TA^2Y_i}
	\right\}^2
\]\[
	= \frac{1}{\sqrt{n}} \sum_{j=1}^p \sum_{k=1}^p 
	\{ \sqrt{n}(A\!-\!A_n)(A\!+\!A_n) \}^{(j,j)} \{ \sqrt{n}(A\!-\!A_n)(A\!+\!A_n)\}^{(k,k)}
	\frac{1}{n} \sum_{i=1}^n \left( \frac{ Y_i^{(j)} Y_i^{(k)} }{Y_i^TA^2Y_i}	\right)^2,
\]
which converges to zero in probability. For $\mathcal{S}_2$, we obtain
\[
	\mathcal{S}_2 = 
	\frac{1}{\sqrt{n}}\sum_{i=1}^n\frac{ \{ Y_i^T(A-A_n)(A_n-A)Y_i \}\, A_nY_iY_i^TA_n}{Y_i^TA^2Y_iY_i^TA^2Y_i}
\]\[
	= 
	\frac{1}{\sqrt{n}}\sum_{j=1}^p
	- \{ a_j^{-1} \sqrt{n}(A_n-A)^{(j,j)}\}^2 \, 
	A_n A^{-1} \left( \frac{1}{n}\sum_{i=1}^n (a_j Y_i^{(j)})^2\frac{AY_iY_i^TA}{(Y_i^TA^2Y_i)^2} \right) A^{-1}A_n
\]\[
	\cid 0 \cdot \sum_{j=1}^p -(Z^{(j,j)}/a_j)^2 \Gamma_j,
\]
where we have again used Lemma A\ref{applem:1}. Similar calculations yield that $\mathcal{S}_3 = o_P(1)$ and $\mathcal{S}_4 = o_P(1)$ as $n \to \infty$. Note that, although we have treated $\mathcal{T}_{1,a}$ and $\mathcal{L}$ individually, they converge in fact jointly. Both are essentially linear functions of $\sqrt{n} (A_n - A)$. The proof of Theorem \ref{th:1} is complete.
\end{proof}
%
%
%
%

\begin{proof}[Proof of Corollary \ref{cor:1}]
As in Theorem \ref{th:1}, let $X_0 = A(X-t)$. Then $X_0 \sim F_0 \in \Ee_2(0,V_0)$. Since $V_0$ has equal diagonal elements, its eigenvalue decomposition is given by $V_0 = U \Lambda U^T$, where
\be \label{eq:eigen}
	 U	= \frac{1}{\sqrt{2}}
	\begin{pmatrix}
			1&1\\
		-1&1
	\end{pmatrix},
	\qquad 
	\Lambda =
	\begin{pmatrix}
		\lambda_1&0\\
		0&\lambda_2
	\end{pmatrix}
	= c
	\begin{pmatrix}
		1-\rho&0\\
		0&1+\rho
	\end{pmatrix}.
\ee
for some $c > 0$. Hence, by Proposition 1 of \citet{duerre:vogel:fried:2015}, we have
\[
	S(F_0,0)=
	\begin{pmatrix}
		1/2  & \delta \\
		\delta & 1/2
	\end{pmatrix}
\]
with $\delta = (1-\sqrt{1-\rho^2})/(2\rho)$ if $\rho \neq 0$ and $\delta = 0$ otherwise, and hence 
\be \label{eq:partI}
	A^{-1} Z S(F_0, 0) + S(F_0, 0) A^{-1} Z =
	\begin{pmatrix}
		Z_1/a_1 										& (Z_1/a_1 + Z_2/a_2) \delta \\
		(Z_1/a_1 + Z_2/a_2) \delta  & Z_2/a_2
	\end{pmatrix}.
\ee
To compute the remaining part $- 2 \sum_{j=1}^2  (Z_j/a_j) \Gamma_j$, we have to evaluate the integrals $\Gamma_j$, $j = 1,2$. Towards this end, we write $X_0 = U \Lambda^{1/2} Y$, where $U$ and $\Lambda$ are as in (\ref{eq:eigen}) and $Y$ has a spherical distribution, and consider the matrix
\[
	W = 
	E\left[
		\vec\left\{ \frac{X_0 X_0^T}{X_0^T X_0} \right\} \vec\left\{ \frac{X_0 X_0^T}{X_0^T X_0} \right\}^T
	\right]
\]\[
	\qquad =
	(U\otimes U) 
	E
	\left[
		\vec\left\{ \frac{ \Lambda^{1/2} Y Y^T \Lambda^{1/2} }{Y^T \Lambda Y} \right\} 
		\vec\left\{ \frac{ \Lambda^{1/2} Y Y^T \Lambda^{1/2} }{Y^T \Lambda Y} \right\}^T
	\right] 
	(U\otimes U)^T
\]
The expectation on the right-hand side is independent of the elliptical generator $g$ and is given as an explicit function of $\lambda_1$ and $\lambda_2$ in the proof of Proposition 2(3) in \citet{duerre:vogel:fried:2015}. Plugging in our specific forms of $\Lambda$ and $U$, cf.~(\ref{eq:eigen}), we obtain
\[
	W = 
	\begin{pmatrix}
	 \alpha & \beta  	& \beta  	& \gamma  \\
	 \beta  & \gamma  & \gamma  & \beta  \\
	 \beta  & \gamma  & \gamma  & \beta  \\
	 \gamma & \beta  	& \beta  	& \alpha 
	\end{pmatrix}
\]
with
\[
	\alpha = \frac{ \sqrt{1-\rho^2} + 2\rho^2 - 1}{4\rho^2}, 
	\quad
	\beta = \frac{1-\sqrt{1-\rho^2}}{4\rho} = \delta/2,
	\quad
	\gamma = \frac{1 - \sqrt{1-\rho^2}}{4\rho^2}
\]
if $\rho \neq 0$, and $\alpha = 3/8$, $\beta = 0$, $\gamma = 1/8$ if $\rho = 0$.
Since $W$ contains $\Gamma_1$ as upper diagonal block and $\Gamma_2$ as lower diagonal block, we obtain
\be \label{eq:partII}
	-2\left(\frac{Z_1}{a_1} \Gamma_1 + \frac{Z_2}{a_2} \Gamma_2 \right) = 
	-2 
	\begin{pmatrix}
		\frac{Z_1}{a_1}\alpha + \frac{Z_2}{a_2}\gamma 				& \left(\frac{Z_1}{a_1} + \frac{Z_2}{a_2}\right) \beta \\
		\left(\frac{Z_1}{a_1} + \frac{Z_2}{a_2}\right) \beta  & \frac{Z_1}{a_1}\gamma + \frac{Z_2}{a_2}\alpha
	\end{pmatrix}.
\ee
Putting (\ref{eq:partI}) and (\ref{eq:partII}) together, we finally arrive at
\[
	\Xi_2 = 
	\begin{pmatrix}
	Z_1/a_1 - Z_2/a_2 	&		0								\\
	0										&   Z_2/a_2 - Z_1/a_1 \\
	\end{pmatrix}
\] 
which completes the proof of Corollary \ref{cor:1}.
\end{proof}
%
%
%
%
%
%

For the proof of Theorem \ref{th:2} we require a slight generalization of the delta method.
\begin{applem}\label{applem:2}
Let $(U_n)_{n\in \N}$ be a series of $p$-dimensional random vectors and $(a_n)_{n\in \N}$ a sequence of real numbers such that $a_n\rightarrow \infty$ as $n \to \infty$ and
\begin{enumerate}[(I)]
\item 
	$a_n(U_n-u)=O_p(1)$ \ as $n \to \infty$ for some $u \in \R^p$. Let furthermore
\item 
	$h : \R^p\rightarrow \R$ be continuously differentiable at $u = (u_1, \ldots, u_p)^T$ with 
	$\frac{\partial h(u)}{\partial u_i} = 0$ for all $i \in I$ for some subset $I \subset \{1,\ldots,p\}$, and
\item 
	$a_n[U_n-u]_{I^C} \cid \Psi$, where $[U_n-u]_{I^C}$ denotes the random vector obtained from $U_n-u$ by deleting all components in $I$.
\end{enumerate}
Then $a_n(h(U_n)-h(u)) \cid  [h'(u)]_{I^C} \Psi$.
\end{applem}
If $I=\emptyset$, Lemma A\ref{applem:2} boils down to the usual delta method. 
If some components of $h'(u)$ are zero (which are gathered in the index set $I$), it suffices to ensure the joint convergence of the remaining components of $a_n(U_n-u)$ and the boundedness in probability of $a_n(U_n -u)$ to conclude the convergence of $a_n(h(U_n)-h(u))$.
\begin{proof}[Proof of Lemma \ref{applem:2}]
The proof is similar to the proof of Lemma 5.3.2. in \cite{Bickel2001}. 
Since $h$ is continuously differentiable, for every $\epsilon>0$ there exists a $\delta>0$ such that 
\begin{align}\label{ab1}
|u-v|\leq \delta \ \Rightarrow \ |h(v)-h(u)-h'(u)(v-u)|\leq \epsilon |v-u|.
\end{align}
Condition (I) implies that $U_n \cip u$, i.e., $P(|U_n-u|\leq \delta)\rightarrow 1$.
Thus using (\ref{ab1}), we have for every $\epsilon>0$ that
$P(|h(U_n)-h(u)- h'(u)(U_n-u)| \leq \epsilon |U_n-u|) \ \rightarrow \ 1$
which implies 
$a_n(h(U_n)-h(u)-h'(u) (U_n-u)) = o_p(|a_n(U_n-u)|) = o_p(1)$. The latter may be re-written as
\begin{align*}
a_n(h(U_n)-h(u))=a_n h'(u) (U_n-u)+o_p(1),
\end{align*}
and the result follows by Conditions (II) and (III) and Slutsky's lemma.
\end{proof}

%
%
%
%
%
%

\begin{proof}[Proof of Theorem \ref{th:2}]
We write
\[
		\sqrt{n}(\hat{\rho}_{\sigma,n}-\rho)
		=
		\sqrt{n}\left( \gamma \{ \vec S_n(\X_n A_n, t_n(\cdot)) \}  - \gamma\{ \vec S(F_0,0) \}   \right),
\]
where $F_0$ is, as in Theorem \ref{th:1}, the distribution of $X_0 = A(X-t)$, and $\gamma:\R^4 \to \R$ is the function that maps 
the (vectorized) two-dimensional spatial sign covariance matrix of an elliptical distribution to the corresponding generalized correlation coefficient. The function $\gamma$ is given by (\ref{eq:rho_n}). Its derivative $\gamma'$ is computed in the proof of Proposition 5 in \citet{duerre:vogel:fried:2015}. Since $F_0$ has equal marginal scales, i.e., $a = 1$, we have
\[
	\gamma'\{ \vec S(F_0,0) \} = 
	\begin{pmatrix}
			0	&	0	&  2 \sqrt{1-\rho^2}(1+ \sqrt{1-\rho^2}) &  0
\end{pmatrix}.
\]
We further decompose
\be \label{eq:srv} 
	\sqrt{n}\vec \left( S_n(\X_n A_n, t_n(\cdot)) -  S(F_0,0)  \right)
\ee\[
	= \ \sqrt{n}\vec \left( S_n(\X_n A_n, t_n(\cdot)) - S_n(\X_n A, A t) \right) 
	\ + \ \sqrt{n}\vec \left( S_n(\X_n A, A t)					- S(F_0,0) 				 \right),
\]
where we call the two terms on the right hand side $\mathcal{T}_1$ and $\mathcal{T}_2$. We deduce two things: First, 
\[
	\sqrt{n}\vec \left( S_n(\X_n A_n, t_n(\cdot)) -  S(F_0,0)  \right) = O_p(1)   \qquad \mbox{ as } n \to \infty,
\]
since $\mathcal{T}_1 \cid \Xi_2$ by Theorem \ref{th:1}, and $\mathcal{T}_2$ converges in distribution as a corollary of the central limit theorem (or as a special case of Proposition 2 in \citet{duerre:vogel:fried:2015}). Second, the third component of (\ref{eq:srv}) converges in distribution to the same limit as $\mathcal{T}_2^{(3)}$,
since $\mathcal{T}_1^{(3)}$ converges to zero in probability by Corollary \ref{cor:1}. Here we use $(\,\cdot\,)^{(3)}$ to denote the third component of a vector.
The asymptotic distribution of $\mathcal{T}_2$ is given by Proposition 2 in \citet{duerre:vogel:fried:2015}. Making use of the particular structure of $V_0$, i.e., equal diagonal elements, cf.~(\ref{eq:eigen}), we obtain  $\mathcal{T}_2^{(3)} \cid N(0,w)$ with $w = (\sqrt{1-\rho^2} + \rho^2 - 1)/(2\rho)^2$ if $\rho \neq 0$ and $w = 1/8$ if $\rho = 0$. Applying Lemma A\ref{applem:2}
with $\gamma$ in the role of $h$, and $I^C = \{3\}$, we obtain
\[
		\sqrt{n}(\hat{\rho}_{\sigma,n}-\rho)
		=	
		[\gamma'\{ \vec S(F_0,0) \}]_{(1,3)} \cdot N(0,w) 
		=
		N(0, (1-\rho^2)^2+(1-\rho^2)^{3/2}).
\]
Note that $\gamma'(\cdot)$ is a $1\times 4$ matrix.
The proof of Theorem \ref{th:2} is complete.
\end{proof}

%
%
%
%
%
%

\begin{proof}[Proof of Corollary \ref{cor:z-trafo}]
By the delta method, the function $h$ has to satisfy
\begin{align}\label{eq:derivative}
	|h'(x)|=\{ (1-x^2)^2+(1-x^2)^{3/2}\}^{-1/2}.
\end{align}
The function $h$ given in Corollary \ref{cor:z-trafo} fulfills this requirement and is further strictly increasing and odd.
To find the antiderivative of (\ref{eq:derivative}), we have used the compute algebra system \citet{maxima}.
Substituting $z=1-\sqrt{1-x^2}$ yields the integral 
$\int \{\sqrt{(1-z)z}(2-z)\}^{-1} dz$, for which Maxima gives the primite $2^{-1/2} \arcsin ( (3z-2)/|z-2| )$.
\end{proof}


{\small

}

\end{document}